\documentclass[10pt,conference]{IEEEtran}
\usepackage[thmmarks]{ntheorem}
\usepackage{amssymb,amsmath,epsfig}

\newtheorem{definition}{Definition}
\newtheorem{theorem}{Theorem}[section]
\newtheorem{lemma}{Lemma}[section]

\def\dE{{\mathbb{E}}}

\def\dR{{\mathbb{R}}}

\newcommand{\mat}{\boldsymbol}

\begin{document}

\title{Average Capacity Analysis of Continuous-Time Frequency-Selective Rayleigh Fading 
Channels with Correlated Scattering Using Majorization}

\author{
\authorblockN{Eduard Jorswieck and Martin Mittelbach}
\authorblockA{Chair of Communications Theory, Communications Laboratory \\ Dresden University of Technology,  01062 Dresden, Germany\\ Email: \{jorswieck,mittelbach\}@ifn.et.tu-dresden.de}
}

\maketitle

\begin{abstract}
Correlated scattering occurs naturally in frequency-selective fading channels and its impact on the performance needs to be understood. In particular, we answer the question whether the uncorrelated scattering model leads to an optimistic or pessimistic estimation of the actual average capacity. In the paper, we use majorization for functions to show that the average rate with perfectly informed receiver is largest for uncorrelated scattering if the transmitter is uninformed. If the transmitter knows the channel statistics, it can exploit this knowledge. We show that for small SNR, the behavior is opposite, uncorrelated scattering leads to a lower bound on the average capacity. Finally, we provide an example of the theoretical results for an attenuated Ornstein-Uhlenbeck process including illustrations.   
\end{abstract}

\section{Introduction}

The ergodic capacity of a single-user multipath fading channel with slow fading is well known for the case when the receiver has perfect channel state information (CSI) \cite{Ozarow1994a,Biglieri1998,Xiang2003}. The resource allocation for such channels - discrete and continuous time - is studied extensively. Bit and power loading as well as rate adaptation for single- \cite{Scaglione1999,Wong1999} and multi-user systems \cite{Rohling2005,Wunder2007} is performed under different quality-of-service (QoS) requirements and under different assumptions on the channel state information (CSI) at the transmitter. The single-user case with perfect CSI at the transmitter and receiver leads to spectral water filling power allocation \cite{Cover1991}. Adaptation to long-term CSI is proposed in \cite{Yao05} under average and outage QoS constraints. The multi-antenna multi-carrier channel is analyzed from an information-theoretic 
perspective in \cite{Bolcskei2002}.

Obviously, the average achievable rate depends not only on the CSI but also on the channel statistics. Often, an uncorrelated scattering channel is assumed. However, this assumption does mostly not apply to ultra-wideband (UWB) channels. Furthermore, correlation occurs if transceiver filters are taken into account, even for an uncorrelated scattering channel \cite{Kafedziski2008}.

Recently, the achievable average rate for a single-user channel with correlated scattering was studied for the single antenna case in \cite{Muller2008,Kafedziski2008,Mittelbach2007} and the multiple antenna case in \cite{Intarapanich2004}. The results indicate that the tap correlation decreases the performance if no CSI is available at the transmitter. If CSI is available the behavior depends on the signal-to-noise ratio (SNR) whether the performance is increased (low SNR) or decreased (high SNR) by tap correlation \cite{Kafedziski2008,Mittelbach2008}. 

The main contribution of this paper is the non-trivial extension of the results from \cite{Mittelbach2008} to the continuous-time case. A different notion of majorization for function is used. The theory is illustrated by a concrete example using the attenuated Ornstein-Uhlenbeck process, i.e., exponentially decaying power and correlation. 

Note that majorization for functions is applied in communication theory before in \cite{Chuah2002}. There the order is used to compare spectra of eigenvalues of Wishart matrices in the context of MIMO systems. In the current paper, we use the order to compare correlation scenarios of frequency selective channels. 

\section{Channel model and capacity formulas}
\label{sec:cm}

We consider a single-user single-antenna frequency selective slowly fading channel, continuous in time and frequency domain. We assume an average power constraint $P$ on the channel input and perfect CSI at the receiver. The noise at the receiver is additive white Gaussian with power spectral density $N_0$. The channel model employed below is an extended version of \cite{Bolcskei2002} and is described in \cite{Mittelbach2007} in detail. 

Let $(X_{\tau})$, $(Y_{\tau})$, $\tau\in\dR$, be real i.i.d. second order Gaussian processes with zero mean and continuous covariance function $R$ with $\int_{-\infty}^{\infty} \int_{-\infty}^{\infty} R(\tau,\tau') d\tau d\tau' < \infty$. By the complex random process 
\begin{eqnarray}
	H = (H_{\tau})=(X_{\tau}+jY_{\tau}),\;\tau\in\dR, \label{eq:channelmodel}
\end{eqnarray}
we model the continuous-time channel impulse response (CIR) of a blockfading multipath channel in lowpass-equivalent form with Rayleigh-distributed magnitudes. The mean energy contained in $(H_{\tau})$ ist given by the constant $c=2\int_{-\infty}^{\infty}R(\tau,\tau)\,\mathrm{d}\tau$, which we use later for normalization.
Note that uncorrelated scattering can be modeled by choosing $R(\tau,\tau') = g(\tau) g(\tau') \delta(\tau-\tau'), \tau,\tau' \in \dR$, with $g$ satisfying $\int_{-\infty}^{\infty} g^2(\tau) d\tau < \infty$ and $\delta$ being the Dirac delta distribution. 

The average rate within the frequency band $(-W/2,W/2)$ is calculated in [nats/s] by \cite{Biglieri1998,Ozarow1994a}
\begin{eqnarray}
	C(\hat{H},\rho,\hat{p},W) = \dE \left[ \int_{-W/2}^{W/2} \log \left( 1 + \rho \hat{p}(f) |\hat{H}_f|^2 \right) df \right] \label{eq:standard}, 
\end{eqnarray}
where $\hat{H} = (\hat{H}_f) = \left( \int_{-\infty}^{\infty} e^{-j 2 \pi f \tau} H_\tau d\tau \right), f \in \dR$, is the Fourier transform of the process $(H_\tau)$. 

The function $f \mapsto \hat{p}(f), f \in (-W/2,W/2)$, is the spectral power allocation function with $\hat{p}(f)\geq 0$ and with power constraint $\int_{-W/2}^{W/2} \hat{p}(f) df = 1$. In later calculations we reasonably assume $\hat{p}$ to be continuous. By $\rho = \frac{P}{N_0 W}$ we denote the average SNR. 

The type of integral to calculate $\hat{H}$ and $C(\hat{H},\rho,\hat{p},W)$ is a stochastic Riemann-integral. The assumptions made for the channel model, particularily for the covariance function $R$, ensure the existence of all involved quantities as shown in \cite{Mittelbach2007}. To evaluate (\ref{eq:standard}) we need to calculate these integrals. Fortunately, we are allowed to exchange expectation operator and integration, which is a property of the stochastic Riemann-integral. The result is derived in \cite{Mittelbach2007} and can be rewritten in integral form as 
\begin{eqnarray}
	C(\hat{\sigma},\rho,\hat{p},W) = \int_{-W/2}^{W/2} \dE_z \left[ \log \left( 1 + \rho \hat{p}(f) \hat{\sigma}(f) z \right) \right] df, \label{eq:1tr}
\end{eqnarray}
where $z$ is an exponentially distributed random variable with expected value one\footnote{In the rest of the paper we use the variable $z$ in this meaning.} and 
\begin{eqnarray}
	\hat{\sigma}(f) & = & \dE \left[ |\hat{H}_f|^2 \right] \nonumber \\ 
	& = & 2 \int_{-\infty}^\infty \int_{-\infty}^\infty R(\tau,\tau') \cos(2\pi(\tau-\tau')f) d\tau d\tau' \label{eq:sigmahat}
\end{eqnarray}
for $f \in (-W/2,W/2)$. We explicitly note that $C(\hat{\sigma},\rho,\hat{p},W)$ is finite for all valid parameters. Note that $\hat{\sigma}$ is continuous and differentiable, which follows from the properties of $R$ and from the First Fundamental Theorem of Calculus. Furthermore, $\hat{\sigma}(f) \geq 0$ and $\int_{-W/2}^{W/2} \hat{\sigma}(f) df < \infty$ due to Parseval's Theorem. Further note that $\hat{\sigma}$ is constant in case of uncorrelated scattering. Subsequent analysis requires to rewrite (\ref{eq:1tr}) as
\begin{eqnarray}
	C(\sigma,\rho,p,W) = W \int_0^1 \dE_z \left[ \log \left( 1 + \rho \sigma(f) p(f) z \right) \right] df \label{eq:Crp}
\end{eqnarray}
with the scaled and shifted functions $\sigma(f) = \hat{\sigma}(W(f-1/2))$, $p(f) = \hat{p}(W(f-1/2))$, $f \in (0,1)$, with power constraint $\int_0^1 p(f) df = \frac{1}{W}$. Subsequently, we want to refer to $\hat{\sigma}$ or $\sigma$ as the spectral fading variance (function). 

\section{Preliminaries}
\label{sec:II}

We review necessary basic definitions from \cite{Chan1987} and results from \cite{Ryff1967}. 
Let $x$ be a real measurable function on $(0,1)$. The \underline{distribution function} of $x$ is given by
$$d_x(s) = \mu \left( \{ t: x(t)>s \} \right), s \in \dR,$$ where $\mu$ is the Lebesgue measure. 
The distribution function is nonincreasing and right-continuous. One says that two functions $x,y$ are equivalent in distribution if $d_x = d_y$. 
The right-continuous inverse of $d_x$ is defined by 
$$x^*(t) = \inf \{ s: d_x(s) \leq t \}, \quad t \in (0,1).$$ It is nonincreasing and is called the \underline{decreasing rearrangement} of $x$. 
The functions $x,x^*$ are equally integrable or (non-integrable) and their integrals are related by $\int_0^1 x^*(t) dt = \int_0^1 x(t) dt$ and $\int_0^s x^*(t) dt \geq \int_0^s x(t) dt$. We observe that non-negativity and continuity of $x$ implies non-negativity and continuity of $x^*$. 

Assuming $x \in L_1$ where $L_1 = L_1(0,1)$ is the space of integrable functions on $(0,1)$, \cite{Ryff1970} showed that the original function $x$ can be recovered from $x^*$ with $x(t) = x^*(\phi(t))$ by
\begin{eqnarray}
	\phi(s) = \mu \{ t : x(t)>x(s) \} + \mu \{ t \leq s : x(t)= x(s) \}. \label{eq:getback}
\end{eqnarray}

Next, we define the partial order on functions $x,y \in L_1$. The following definitions can be found in  \cite[Def. 1.2]{Chan1987}. Majorization for functions is also discussed in \cite[Sec. 2.1.3]{Jorswieck2007c}. 
\begin{definition}
	Let $x,y \in L_1$. We say that $x$ \underline{majorizes} $y$ and write $x \succeq y$ if 
	\begin{eqnarray}
		\int_0^s x^*(t) dt \geq \int_0^s y^*(t) dt \label{eq:d1}
	\end{eqnarray}
	for all $s$ with $0 \leq s < 1$ and 
	\begin{eqnarray}
		\int_0^1 x^*(t) dt = \int_0^1 y^*(t) dt \label{eq:d2}. 
	\end{eqnarray}
\end{definition}
\textit{Remark:} One may conclude there always exists an $\epsilon>0$ such that for continuous $x,y$ we have $x^*(t) \geq y^*(t)$ for $t \in (0,\epsilon)$. This is not true in general. However, if we additionally require $$x^*(0+) > y^*(0+),$$ where $x^*(0+) = \lim_{t \downarrow 0}x^*(t)$, then there exists an $\epsilon>0$, such that $x^*(t)>y^*(t)$ for $t \in (0,\epsilon)$. This additional requirement we will use in Theorem \ref{theo:2}. 

The order-preserving functionals with respect to majorization are called Schur-convex and Schur-concave. 
\begin{definition}
	A real-valued functional $\phi$ defined on $\mathcal{A} \subset L_1$ is said to be \underline{Schur-convex} on $\mathcal{A}$ if for all $x,y \in \mathcal{A}$ with $x \succeq y$ it follows that $\phi(x) \geq \phi(y)$. 
	
	If $-\phi$ is Schur-convex on $\mathcal{A}$ then $\phi$ is \underline{Schur-concave} on $\mathcal{A}$. 
\end{definition}
The following result is used later. It can be found in \cite{Ryff1967} and was first proved in \cite{Hardy1929}. The corresponding result on vectors is given in \cite[Prop. 3.C.1]{Marshall1979}.
\begin{lemma}
	Let $x,y \in L_1$ with $x \succeq y$. If $g$ is a real-valued, concave function on an interval $I \subset \dR$ such that ${g \circ x, g \circ y \in L_1}$, then
	\begin{eqnarray}
		\int_0^1 (g \circ x)(t) dt \leq \int_0^1 (g \circ y)(t) dt \label{eq:hardy}.
	\end{eqnarray}
	\label{lem:xy}
\end{lemma}
In order to compare different correlation scenarios, we can use majorization for functions.
\begin{definition}
We say that a channel with correlation function $R_1$ and corresponding spectral fading variance $\sigma_1$ is \underline{more correlated} than a channel with correlation function $R_2$ and corresponding spectral fading variance $\sigma_2$ if $\sigma_1 \succeq \sigma_2$. 
\end{definition}

\section{Average capacity characterizations}

In this section, we present our main results. Two scenarios are studied, namely the case in which the transmitter has no CSI and the case in which it knows the spectral fading variance function $\sigma$. 

\subsection{No CSI at transmitter}

If the transmitter has no CSI, the most reasonable strategy is to apply equal power allocation, i.e. $p(f) =\frac{1}{W}, f \in (0,1)$. The average rate is then given by
\begin{eqnarray}
	C_{\mathrm{no}}(\sigma,\rho,W) = W  \int_0^1 \dE_z \left[ \log \left( 1 + \frac{\rho}{W} \sigma(f) z \right) \right] df \label{eq:noCSI}.  
\end{eqnarray}

\begin{theorem}
\label{theo:1}
	The functional $C_{\mathrm{no}}$ given in (\ref{eq:noCSI}) with no CSI at the transmitter is Schur-concave with respect to the spectral fading variance $\sigma$ for all $\frac{\rho}{W}\geq 0$, i.e., $\sigma_1 \succeq \sigma_2$ implies $C_{\mathrm{no}}(\sigma_1) \leq C_{\mathrm{no}}(\sigma_2)$. 
\end{theorem}

\begin{proof}
	We apply Lemma \ref{lem:xy} to prove the theorem. For all $\alpha \geq 0$, the function $s \mapsto \log\left( 1+\alpha s\right)$ is concave on $\dR_+$ and thus $s \mapsto \log \left( 1+\frac{\rho}{W} s z(w) \right)$ is concave on $\dR_+$ for each realization $z(w)$ of $z$. Since the Lebesgue-integral is monotonic (and thus the expectation operator), $$s \mapsto \phi(s) = \dE_z \left[ \log \left( 1 + \frac{\rho}{W} s z \right) \right]$$ is concave on $\dR_+$. Since $C_{\mathrm{no}}(\sigma,\rho,W) = W \int_0^1 (\phi \circ \sigma)(f) df$ is finite under the assumptions made, as noted in Section \ref{sec:II} and proved in \cite{Mittelbach2007}, we can apply the Lemma \ref{lem:1} which completes the proof. 
\end{proof}

\textit{Remark:} Assume two channels with spectral fading variance $\sigma_1$ and $\sigma_2$ and further assume $\sigma_1(f)>0$, $\sigma_2(f)>0$, $f \in (0,1)$. For high average SNR, i.e., high values of $\rho$, the average rate $C_{\mathrm{no}}(\sigma,\rho,W)$ can be approximated by committing the $1$ in the logarithm in (\ref{eq:noCSI}) to obtain 
\begin{eqnarray}
	\tilde{C}_{\mathrm{no}}(\sigma,\rho,W) = W \int_0^1 \dE_z \left[ \log \left( \frac{\rho}{W} \sigma(f) z \right) \right] df.\label{eq:Cno}
\end{eqnarray}
We use this approximation to calculate the difference of the average rates for fading variances $\sigma_1$ and $\sigma_2$ for high SNR and obtain
\begin{eqnarray}
	\Delta \tilde{C}_{\mathrm{no}}(\sigma_1,\sigma_2) & = & \tilde{C}_{\mathrm{no}}(\sigma_1,\rho,W) - \tilde{C}_{\mathrm{no}}(\sigma_2,\rho,W) \nonumber \\ & = & W \int_0^1 \log \left( \frac{ \sigma_1(f)}{\sigma_2(f)} \right) df \label{eq:diff}. 
\end{eqnarray}
The equation is the analogue to equation (19) in \cite{Mittelbach2008}. 

\subsection{Partial CSI at the transmitter}

If the transmitter knows the spectral fading variance, it can adapt the power allocation function $p$ accordingly. The derivation of the optimal power allocation is similar to \cite{Goldsmith1997} and we obtain
\begin{eqnarray}
	& C_{\mathrm{part}}(\sigma,\rho,W) = \max\limits_{p} W \int\limits_0^1 \dE_z \left[ \log \left( 1 + \rho \sigma(f) p(f) z \right) \right] df \nonumber \\
	& \mathrm{s.t} \quad p(f) \geq 0, \quad \int_0^1 p(f) df \leq \frac{1}{W}, \quad f \in (0,1) \label{eq:mopt} . 
\end{eqnarray}
We easily verify that the constraint set is convex. Furthermore, the functional $p \mapsto W \int_0^1 \dE_z \left[ \log \left( 1 + \rho \sigma(f) p(f) z \right) \right] df$ is strictly concave on the constraint set, which is identically shown as in the proof of Theorem \ref{theo:1}. 

Thus the Karush-Kuhn-Tucker (KKT) conditions are sufficient for a unique global maximum \cite{Boyd2003}. Evaluation of the KKT conditions gives the following characterization of the optimal solution $p_{\mathrm{o}}$:
\begin{eqnarray}
	\dE_z \left[ \frac{ \rho \sigma(f) z }{1 + \rho \sigma(f) p_{\mathrm{o}}(f) z} \right] = \nu - \lambda(f) \nonumber \\
	\nu \geq 0, \; \lambda(f) \geq 0, \; p_{\mathrm{o}}(f) \lambda(f) = 0, \; f \in (0,1), \label{eq:optbed}
\end{eqnarray}
where $\nu$ and $f \mapsto \lambda(f), f \in (0,1)$, are Lagrangian multipliers. Without loss of generality we can assume $\nu > 0$. 

The average capacity with $\sigma$ known at the transmitter is then given by
\begin{eqnarray}
C_{\mathrm{part}}(\sigma,\rho,W) = W \int_0^1 \dE_z \left[ \log\left( 1 + \rho \sigma(f) p_{\mathrm{o}}(f) z \right) \right] df. 
\label{eq:Cpart}
\end{eqnarray}

\textit{Remark:} We can equivalently use the decreasing rearrangement $\sigma^*$ of $\sigma$ to calculate the capacity. We simply replace $\sigma$ by $\sigma^*$ and $p_{\mathrm{o}}$ by $p_{\mathrm{o}}^*$ in (\ref{eq:mopt}), (\ref{eq:optbed}), and (\ref{eq:Cpart}). Then, the optimal power allocation $p_{\mathrm{o}}$ belonging to $\sigma$ can be obtained from the optimal power allocation $p_{\mathrm{o}}^*$ belonging to $\sigma^*$ by $p_{\mathrm{o}}(f) = p_{\mathrm{o}}^*(\phi(f)), f \in (0,1)$, with $\phi$ as in (\ref{eq:getback}). 

In the following we consider the quantity 
\begin{eqnarray}
	\theta(\sigma,\rho,W) = \mu \left( \{ f: p_{\mathrm{o}}(f) >0 \} \right) = \mu \left( \{ f: p_{\mathrm{o}}^*(f) > 0 \} \right) \label{eq:volume},
\end{eqnarray}
i.e., the Lebesgue-measure of the support of $p_{\mathrm{o}}$ and call it the \underline{volume of active frequencies}. 

\begin{lemma}
	For given fading variance $\sigma$ and bandwidth $W$, the function $\rho \mapsto \theta(\sigma,\rho,W)$ is strictly monotonic increasing on $[0,\infty)$. In particular, for given $\epsilon>0$ there exists a unique $\tilde{\rho}(\epsilon)$ such that $\theta(\sigma,\rho,W)<\epsilon$ for all $\rho < \tilde{\rho}(\epsilon)$. 
	\label{lem:1}
\end{lemma}	

\begin{proof}
	We define the function 
	\begin{eqnarray}
	  x \mapsto \psi(x) & = & \dE_z \left[ \frac{z}{1+x z} \right] \nonumber \\ & = & \frac{1}{x} - \frac{e^{1/x} \mathrm{Ei}_1(1/x)}{x^2} \label{eq:defpsi}, \; x \in [0,\infty)
	\end{eqnarray}
	with $\mathrm{Ei}_1(y)=-\mathrm{Ei}(-y)$, where $\mathrm{Ei}$ is the exponential integral \cite[Ch. 5.1]{Abramowitz}. The function is strictly convex and strictly monotonic decreasing with $\psi(0)=1$ and $\lim_{x \rightarrow \infty} \psi(x) = 0$. These properties hold for the function $x \mapsto \frac{z}{1+z x}, x \in [0,\infty)$, for all $z \geq 0$ (verified with first and second derivative) and hold also for $\psi$ due to the monotonicity of the Lebesgue-integral. 
	
	The properties of $\psi$ ensure the existence of the inverse function $\psi^{-1}$, which is defined on $(0,1]$. The inverse $\psi^{-1}$ is strictly convex, strictly monotonic decreasing and $\psi^{-1}(x)>0$ for $x \in (0,1)$. We extend $\psi^{-1}$ by defining $\psi^{-1}(x) = 0$ for $x \in (1,\infty)$. 
	
	For $\alpha>0$ consider now the function $x \mapsto \zeta_\alpha(x) = x \psi^{-1}(\alpha x)$ for $x \in (0,\frac{1}{\alpha}]$ and $\zeta_\alpha(x) = 0$ for $x \in (\frac{1}{\alpha},\infty)$. It is not difficult to show that $\zeta_\alpha$ is strictly monotonic decreasing on $(0,\frac{1}{\alpha}]$ with $\zeta_\alpha(0+) = 1$ and $\zeta_\alpha(\frac{1}{\alpha}) = 0$. 
	
	From the optimality condition in (\ref{eq:optbed}) it follows
	\begin{eqnarray}
		p_{\mathrm{o}}(f) = \zeta_{\nu} \left( \frac{1}{\rho \sigma^*(f)} \right) = \frac{1}{\rho \sigma^*(f)} \psi^{-1} \left( \frac{\nu}{\rho \sigma^*(f)} \right) \label{eq:popt1},
	\end{eqnarray}
	for all $f \in (0,1)$ for which $p_{\mathrm{o}}^*(f)>0$. 
	
	Since $\sigma^*$ is monotonically decreasing on $(0,1)$ it follows that $f \mapsto \zeta_\nu \left( \frac{1}{\rho \sigma^*(f)} \right)$ is monotonically decreasing on $(0,1)$ for constant $\nu >0$ and $\rho >0$. 
	
	If $\zeta_\nu \left( \frac{1}{\rho \sigma^*(1)} \right) > 0$ we define $f^*=1$ and otherwise we define $f^*$ to be the smallest $f$ such that $\nu = \rho \sigma^*(f)$. All frequencies $f < f^*$ are active and all other frequencies are not. Thus, we have $\theta(\sigma,\rho,W) = f^*$. From the monotonicity of $\zeta_\nu$ it easily follows, that $f^*$ increases for increasing $\rho$ if $\nu$ and $\sigma^*$ are fixed. However, since $\nu$ is the parameter guaranteeing the condition $\int_0^1 p^*_{\mathrm{o}}(f) df = \frac{1}{W}$ to hold, it is also a function of $\sigma^*$ and $\rho$. Thus, to show that $f^*$ increases with increasing $\rho$, we finally have to show that $\frac{\nu(\sigma^*,\rho)}{\rho}$ decreases with increasing $\rho$. 
\end{proof}

In case of known $\sigma$ at the transmitter, we obtain the following theorem characterizing the impact of correlation for low SNR. 

\begin{theorem}
	Let $\sigma_1$ and $\sigma_2$ be fading variances with $\sigma_1 \succeq \sigma_2$ and $\sigma_1^*(0+) > \sigma_2^*(0+)$. Then for given bandwidth $W$ there exists a $\tilde{\rho}>0$ such that $$C_{\mathrm{part}}(\sigma_1,\rho,W) \geq C_{\mathrm{part}}(\sigma_2,\rho,W)$$ for all $\rho \leq \tilde{\rho}$. 
	\label{theo:2}
\end{theorem}

\begin{proof}
	Since $\sigma_1$ and $\sigma_2$ are continuous and $\sigma_1(0+)> \sigma_2(0+)$, there is an $\epsilon>0$ such that $\sigma_1^*(f)>\sigma_2^*(f)$ for $f \in (0,\epsilon)$. We take the largest possible $\epsilon$ with this property. According to Lemma \ref{lem:1} there exists $\tilde{\rho}_i(\epsilon)$ such that $\theta(\sigma_i,\rho,W)< \epsilon$ for $\rho< \tilde{\rho}_i(\epsilon)$, $i=1,2$. Since $\sigma_1^*(f)> \sigma_2^*(f)$ for $f \in (0,\epsilon)$, we have $\tilde{\rho}_1(\epsilon)>\tilde{\rho}_2(\epsilon)$. Choosing $\rho \in (0,\tilde{\rho}_2(\epsilon))$ we can rewrite (\ref{eq:Cpart}) as
	\begin{eqnarray}
		& C_{\mathrm{part}}(\sigma,\rho,W) = \max\limits_{p} W \int_0^\epsilon \dE_z \left[ \log \left( 1 + \rho \sigma^*(f) p^*(f) z \right) \right] df \nonumber \\
		& \mathrm{s.t.} \; p^*(f)\geq 0, \; \int_0^1 p^*(f)df \leq \frac{1}{W}, \; f \in (0,1). \nonumber
	\end{eqnarray}
	Again using $\sigma_1^*(f) > \sigma_2^*(f)$ on $(0,\epsilon)$ yields
	\begin{eqnarray}
		\dE \left[ \log \left( 1 + \rho \sigma_1^*(f) p^*(f) z\right) \right] > \dE \left[ \log \left( 1 + \rho \sigma_2^*(f) p^*(f) z\right) \right] \label{eq:stern}
	\end{eqnarray}
	for all $f \in (0,\epsilon)$, $\rho \in (0,\tilde{\rho}_2(\epsilon))$, and valid $p^*$. Applying $\int_0^1 (\cdot) df$ and $\max_p (\cdot)$ on both sides of (\ref{eq:stern}) does not change the order of the inequality which completes the proof. 	
\end{proof}
We remark that for small SNR, the effect of correlation with known fading variance is opposite compared to the no CSI scenario. It can be shown using the same approximation as in (\ref{eq:Cno}) that for large SNR, the optimal power allocation is equal power allocation and the behavior with known fading variance at the transmitter is identical to the no CSI scenario.

\section{Example and illustrations}

\subsection{Definition of channel model}

As in \cite{Mittelbach2007}, we consider an exponentially attenuated Ornstein-Uhlenbeck process to illustrate the results. The channel model in time-domain, given by (\ref{eq:channelmodel}), is then described by the covariance function for the real and imaginary part 
\begin{eqnarray}
	R(\tau,\tau') = c e^{-a|\tau-\tau'|} b e^{-b(\tau+\tau')} \mat{1}_{\{\tau \geq 0\}}(\tau) \mat{1}_{\{\tau'\geq 0\}}(\tau'),
\end{eqnarray}
for $\tau, \tau' \in \dR$, where $c>0$ is the normalization constant introduced in Section \ref{sec:cm}. The covariance function captures an exponential power decay controlled by parameter $b$. In addition, the correlation decays exponentially with delay separation controlled by parameter $a$. 

\subsection{Fading variance specification}

The spectral correlation function given by (\ref{eq:sigmahat}) can be computed in closed form \cite{Mittelbach2007}
\begin{eqnarray}
	\hat{\sigma}_d(f) = \frac{2cd}{d^2 + (2 \pi f)^2}, \; f \in (-W/2,W/2) \label{eq:suo1},
\end{eqnarray}
with the parameter $d = a+b$. Note that the power and correlation decay parameters have the same impact on the fading variance, since they occur in (\ref{eq:suo1}) only as sum. We choose the parameter $c$ such that $\int_{-W/2}^{W/2} \hat{\sigma}(f)df = 1$, i.e., $c = \frac{\pi}{2 \arctan\left( \frac{\pi W}{d} \right)}$. Further, we require the shifted and scaled version of $\hat{\sigma}_d$, given by $\sigma_d(f) = \hat{\sigma}_d(W(f-1/2)), f\in (0,1)$. Due to the symmetry of $\sigma_d$ we easily obtain the decreasing rearrangement 
\begin{eqnarray}
	\sigma_d^*(f) = \frac{\pi d}{\arctan \left( \frac{\pi W}{d} \right)} \cdot \frac{1}{d^2 + \left( \pi W f \right)^2}, \; f\in (0,1). \label{eq:sdf}
\end{eqnarray}
For the considered example, we obtain the following result, which allows us to apply Theorem \ref{theo:1} and \ref{theo:2}.
\begin{lemma}
	Let $\sigma_{d_1}$ and $\sigma_{d_2}$ be the fading variances for parameters $d_1,d_2>0$. If $d_1 < d_2$ then $\sigma_{d_1} \succeq \sigma_{d_2}$. Further, we have $\sigma_{d_1}^*(0+) > \sigma_{d_2}^*(0+)$ and $\sigma_{d_1}^*(f) = \sigma_{d_2}^*(f)$ for 
	\begin{eqnarray}
		f = \frac{1}{\pi W} \sqrt{ d_1 d_2 \cdot \frac{ d_2 A_{d_2}-d_1 A_{d_1}}{d_2 A_{d_1} - d_1 A_{d_2}}}, \label{eq:fschnitt}
	\end{eqnarray}
	where $A_x = \arctan \left( \frac{\pi W}{x} \right)$. 
	\label{lem:hirn}
\end{lemma}

\begin{proof}
	We give here only the sketch of the proof since the derivations are tedious and due to lack of space. 
	First, we compute $\sigma_d^*(0+) = \frac{\pi}{d \arctan \left( \frac{\pi W}{d} \right)}$ and then show that the function $d \mapsto \sigma_d^*(0+), d \in (0,\infty)$, is strictly monotonic decreasing using its first derivative. This implies $\sigma_{d_1}^*(0+) > \sigma_{d_2}^*(0+)$. Equation (\ref{eq:fschnitt}) is directly obtained by calculating the solution of the equation $\sigma_{d_1}^*(f) = \sigma_{d_2}^*(f)$. 
	
	To prove that $\sigma_{d_1} \succeq \sigma_{d_2}$ we show that for all $s \in [0,1)$ the function $\xi_s$ with 
	\begin{eqnarray}
		\xi_s(d) = \int_0^s \sigma_d^*(f) df = \frac{ \arctan \left( \frac{\pi W s}{d} \right)}{W \arctan \left( \frac{\pi W}{d} \right)}, \; d \in (0,\infty), \label{eq:xis}
	\end{eqnarray}
	is monotonically decreasing using its first derivative and that $\xi_1(d) = \frac{1}{W}$. 
\end{proof}

\subsection{Illustrations}

We use $W=1$ in all subsequent simulations. Fig. \ref{fig:sim1} shows the decreasing rearrangements of three spectral fading variances illustrating the results of Lemma \ref{lem:hirn}. 
\begin{figure}[h!]
	\begin{center}
		\includegraphics[width=.7\linewidth]{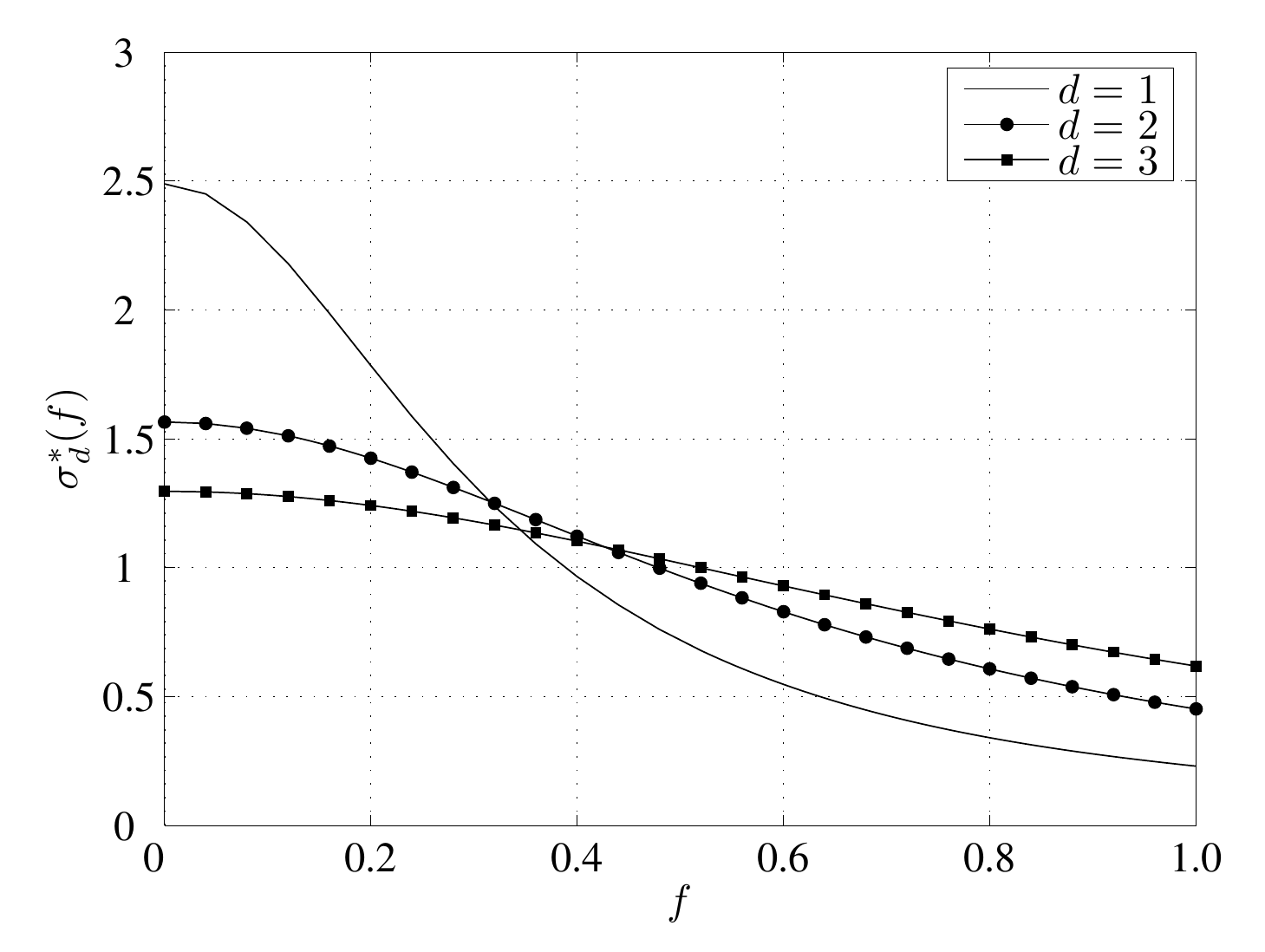}
	\end{center}
	\caption{Decreasing rearrangement of spectral fading variances.}
	\label{fig:sim1}
\end{figure}
We observe that with increasing $d$ the rearranged fading variance is more spread out, which fits well with the notion of majorization. 

\begin{figure}[h!]
	\begin{center}
		\includegraphics[width=.7\linewidth]{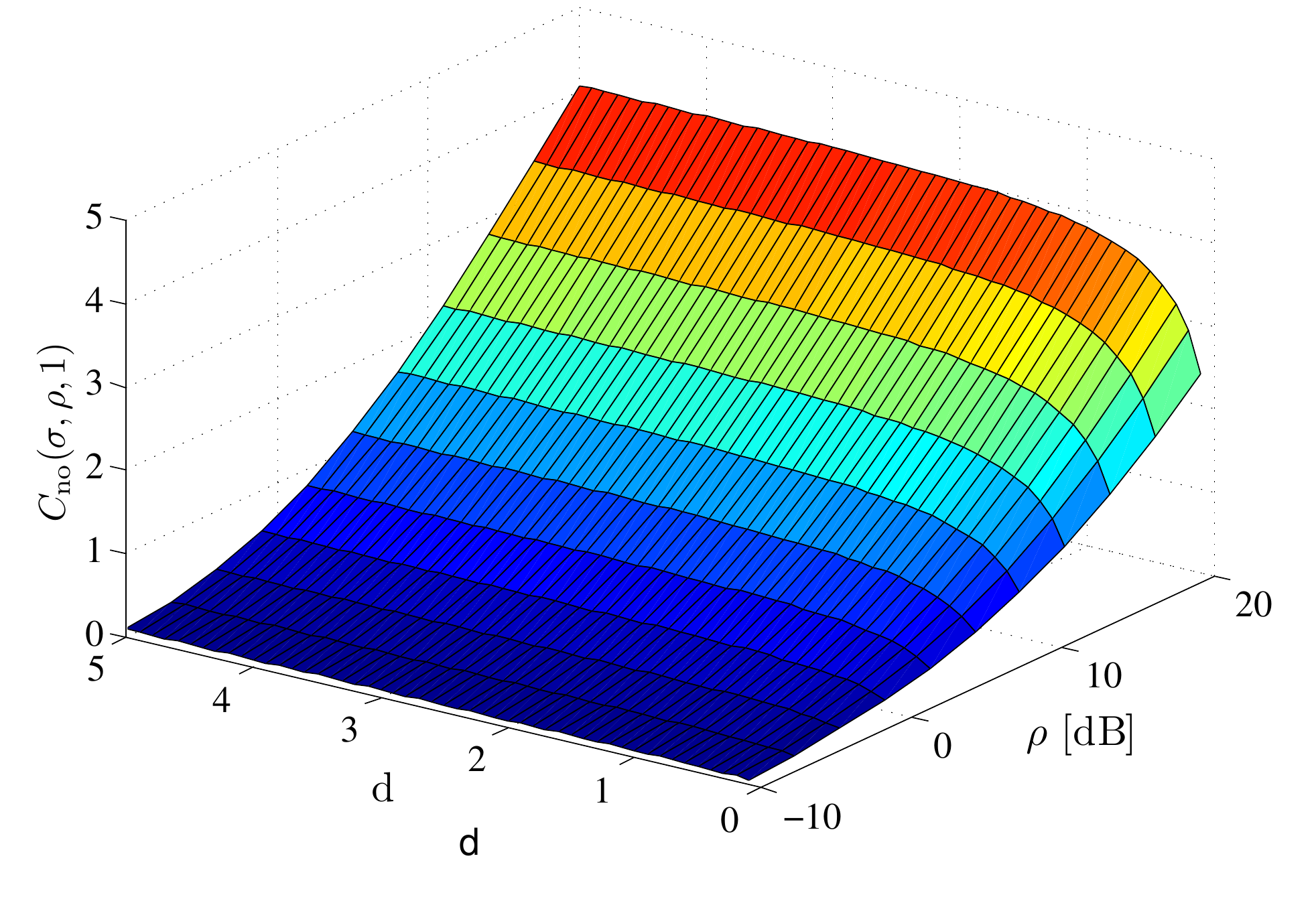}
	\end{center}
	\caption{Average capacity with no CSI at transmitter.}
	\label{sim:2}
\end{figure}

Fig. \ref{sim:2} shows that the average rate $C_{\mathrm{no}}$ with no CSI at the transmitter increases with increasing $d$ for all values of $\rho$, which follows from Lemma \ref{lem:hirn} and Theorem \ref{theo:1}. 

\begin{figure}[h!]
	\begin{center}
		\includegraphics[width=.7\linewidth]{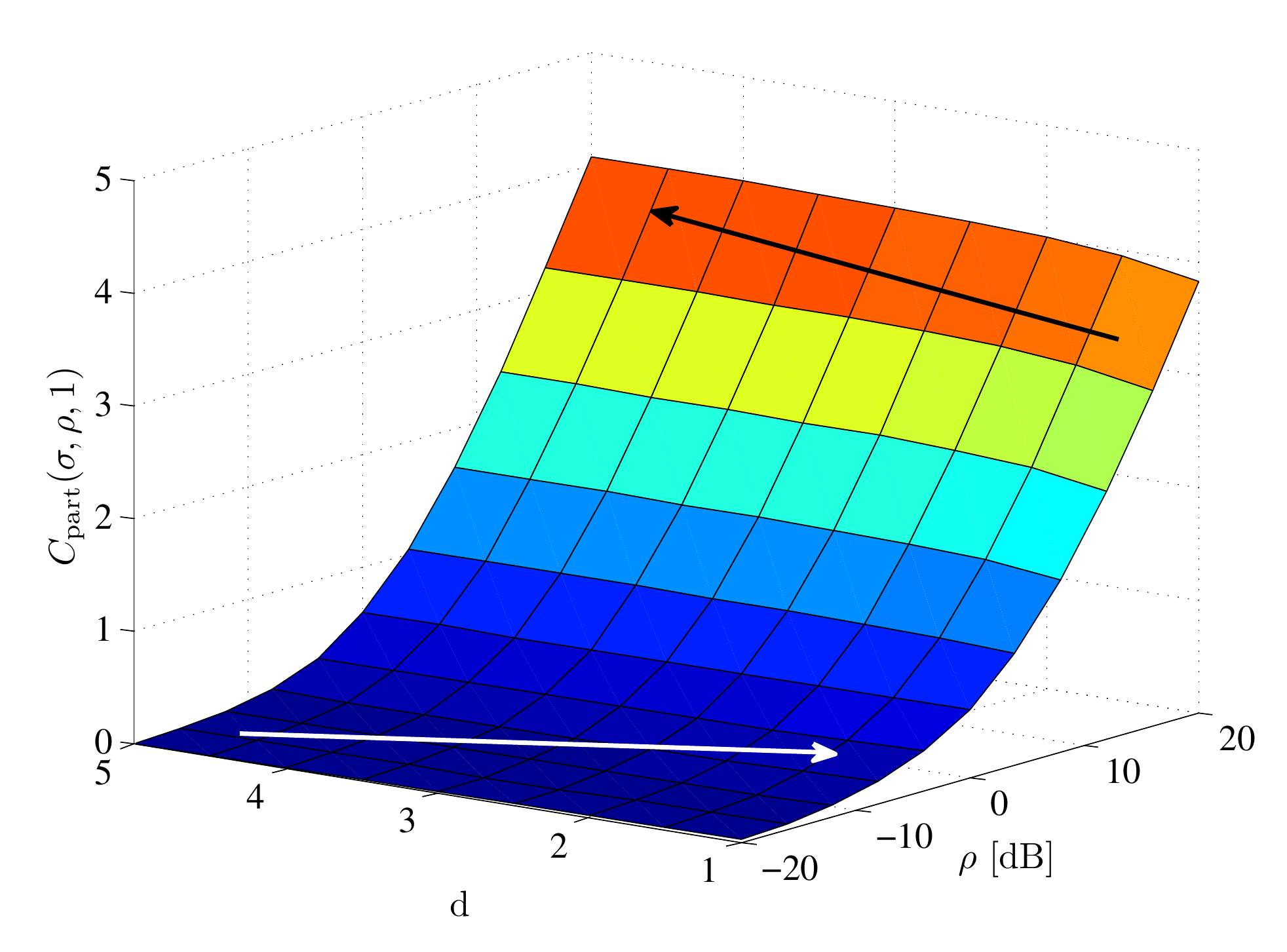}
	\end{center}
	\caption{Average capacity with known fading variance at transmitter.}
	\label{sim:3}
\end{figure}

Finally, Fig. \ref{sim:3} illustrates the result from Theorem \ref{theo:2} and Lemma \ref{lem:1}. The high SNR behaviour is identical to the case with no CSI, for low SNR it is reversed.



\end{document}